\documentclass[12pt]{article}

\usepackage{bbm,fullpage}
\usepackage{bm}
\usepackage{amsmath,amssymb,amsthm,amsmath}
\usepackage{algorithm, pseudocode}
\usepackage{mathrsfs}
\usepackage{enumitem}

\usepackage[titles]{tocloft}
\usepackage{yfonts}
\usepackage{xcolor}
\usepackage[pdftex,                %
    bookmarks         = true,
    bookmarksnumbered = true,
    pdfpagemode       = None,
    pdfstartview      = FitH,
    pdfpagelayout     = SinglePage,
    colorlinks        = true,
   citecolor         =blue,
    urlcolor          = magenta,
    pdfborder         = {0 0 0}
    ]{hyperref}
 
\usepackage{bbm}
\usepackage{amsmath}
\usepackage{easybmat}
\usepackage{multirow,bigdelim}

\usepackage[indexonlyfirst,ucmark,toc]{glossaries}

\definecolor{ao(english)}{rgb}{0.0, 0.5, 0.0}

\newtheorem{definition}{Definition}
\numberwithin{definition}{section}
\newtheorem{theorem}[definition]{Theorem}

\newtheorem{proposition}[definition]{Proposition}
\newtheorem{lemma}[definition]{Lemma}

\newtheorem{example}[definition]{Example}

\def\K{\ensuremath{\mathbb{K}}}
\def\Kbar {\ensuremath{\overline{\mathbb{K}}}}
\DeclareBoldMathCommand{\c}{c}
\DeclareBoldMathCommand{\f}{f}
\DeclareBoldMathCommand{\g}{g}
\DeclareBoldMathCommand{\h}{h}
\DeclareBoldMathCommand{\x}{x}
\DeclareBoldMathCommand{\z}{z}
\DeclareBoldMathCommand{\v}{v}
\DeclareBoldMathCommand{\0}{0}
\DeclareBoldMathCommand{\u}{u}
\DeclareBoldMathCommand{\e}{e}

\def\ZZ {\ensuremath{\mathbb{Z}}}
\def\bmu{\mbox{\boldmath$\mu$}}

\def\b_eta{\mbox{\boldmath$\eta$}}
\def\softO{\ensuremath{{O}{\,\tilde{ }\,}}}
\DeclareBoldMathCommand{\balpha}{\alpha}
\def\jac{\ensuremath{{\rm Jac}}}
\def\diag{\ensuremath{\mathrm{diag}}}

\title{On the complexity of invariant polynomials under the action of
  finite reflection groups}

\author{Thi Xuan Vu\thanks{Department of Mathematics and Statistics,
    UiT, The Arctic University of Norway, Troms\o{}, Norway email:{thi.x.vu@uit.no}}
}
\date{}
\begin{document}

\maketitle 

\begin{abstract}
   Let  $\K[x_1, \dots, x_n]$ be a  multivariate 
 polynomial ring over a field $\mathbb{K}$. Let $(u_1, \dots,
 u_n)$ be a sequence of $n$  algebraically independent elements in
 $\K[x_1, \dots, x_n]$.  Given a polynomial $f$ in
 $\K[u_1, \dots, u_n]$,  a subring of $\K[x_1,
 \dots, x_n]$ generated by the $u_i$'s, we are interested in 
 finding the unique polynomial $f_{\rm  new}$ in $\mathbb{K}[e_1,
 \dots, e_n]$, where $e_1, \dots, e_n$ are new variables, such that  
  $f_{\mathrm{new}}(u_1, \dots, u_n) = f(x_1, \dots, x_n)$. We
  provide an algorithm and analyze its arithmetic complexity to
  compute  $f_{\mathrm{new}}$ knowing $f$ and  $(u_1, \dots, u_n)$.  


\end{abstract}


\section{Introduction}
 Let $\K[x_1, \dots, x_n]$ and $\K[e_1, \dots, e_n]$ be multivariate
 polynomial rings over a field $\K$. Given a polynomial $f\in \K[u_1,
 \dots, u_n]$, where $u_1, \dots, u_n$ are algebraically independent
 in $\K[x_1, \dots, x_n]$, in this paper we  consider the
 problem of finding the  polynomial $f_{\rm  new} \in \K[e_1, \dots,
 e_n]$ such that $$f_{\mathrm{new}}(u_1, \dots, u_n) = f(x_1, \dots,
 x_n).$$  The problem of finding such polynomial appears in many
 application areas, especially in solving polynomial systems invariant
 under the action of finite reflection groups. 
  Without loss of generality, for  $1 \le i \le k \le n$,
we can suppose that $\deg(u_i) \le \deg(u_k)$. We also 
can assume that $\deg(u_i) \le d$ for all $1 \le i \le n$; otherwise,
the polynomial $f_{\rm new}$ is in $\K[e_1, \dots, e_{k-1}]$ for  $k
\in \{1, \dots, n\}$ such that  $\deg(u_k) > d$. 
 
 \paragraph{Motivation.} Given a field $\K$ and an action of a finite
 group $G$ on a $\K$-vector space $V$, we obtain a linear action of
 $G$ on $\K[V]$, the ring of  polynomial functions of $V$ with
 coefficients in $\K$. A polynomial $f$ in $V$ is said to be
 $G$-{invariant} if $\sigma(f) = f$ for all $\sigma \in G$. We denote
 by $\K[V]^G$ the set of all $G$-invariant polynomials in $\K[V]$. The 
  Chevalley-Shephard-Todd Theorem \cite[Chapter~6]{benson1993} states
  that if the order of $G$ is not divisible by the characteristic of
  $\K$, then $\K[V]^G$  is polynomial if and only if $G$ is generated
  by pseudoreflections. 

 In other words, picking a dual basis $(x_1,
  \dots, x_n)$ for $V^*$, there exist $n$ homogeneous polynomials
  $(u_1, \dots, u_n)$ in $\K[x_1, \dots, x_n]$ such that for any
  polynomial $f$ in $\K[x_1, \dots, x_n]^G$, there exists a unique
  polynomial $f_{\rm new} \in \K[e_1, \dots, e_n]$, where $(e_1,
  \dots, e_n)$ are new variables, such that  
 $
  f_{\rm new}(u_1, \dots, u_n) = f(x_1, \dots, x_n). 
 $ Note that for a fixed finite reflection group $G$, sets of
 generators for $\K[x_1,\dots, x_n]^G$ are not unique. For instance,
 when $G$ is the symmetric group ${S}_n$, the invariant ring $\K[x_1,
 \dots, x_n]^{{S}_n}$ is generated by elementary symmetric polynomials
 in $(x_1, \dots, x_n)$. Alternatively, the ring $\K[x_1, \dots,
 x_n]^{{S}_n}$ can also be generated by power sum symmetric
 polynomials.    

 Computationally, the question of finding the unique representation
 $f_{\rm new}$ of $f$ is important, particularly in polynomial system
 solving with invariant polynomials (see
 e.g.~\cite{FLSSV2021, Vu2022}). Finding such a representation
 allows us to represent a $G$-orbit of a 
 fixed point by a single point. Here the orbit of a point $a \in
 \Kbar{}^n$ is the set 
 $\{g(a) \,:\, {\rm  for~all} \, g \in G\}$ and it sometimes is called 
 the $G$-orbit of $a$, where  $\Kbar$ is an  algebraic closure of
 $\K$. Consider an algebraic
 set $W \subset \Kbar{}^n$  which is invariant under the action of
 $G$. We can define the relative orbit $W/G$ whose points are the
 $G$-orbits of points in the  set $W$.

 \begin{example}Consider a sequence of polynomials $(f, h, \ell) =
   (x_1x_2x_3+2, x_1^2+x_2^2+x_3^2-6, x_1+x_2+x_3-2)$ in 
 $\K[x_1, x_2, x_3]^{{S}_3}$. In this case we can 
 see that the zero set $W := V(f, h, \ell)$ of $(f, h, \ell)$ consists of 6
 points. However, we can easily determine it computationally using
 the  representations \[f_{\rm new} = 
 \frac{1}{6}e_1^3-\frac{1}{2}e_1e_2+\frac{1}{3}e_3+2, h_{\rm new} =
 e_2-6, \,  {\rm and} \, \ell_{\rm new} = e_1-2\] of $f$, $h$, and
 $\ell$, respectively,  with a set of generators for  $\K[x_1,  x_2,
 x_3]^{{S}_3}$ being $$(u_1, u_2, u_3) = (x_1+x_2+x_3,
 x_1^2+x_2^2+x_3^2, x_1^3+x_2^3+x_3^3).$$ Here $(e_1, e_2, e_3)$ stand
 for new variables. 

 The solution
 set of $(f_{\rm new}, h_{\rm new}, \ell_{\rm new})$ consists of only
 one point $(2,6,8)$. That is the relative  orbit $W/{S}_3$ contains a
 single point, implying $W$ has only a single orbit consisting of six
 permutations of $(-1,1,2)$. The point $(2,6,8)$ represents the
 ${S}_3$-orbit of $(-1,1,2)$. This result matches  what we get if we
 compute directly the orbit of any element in $W$.  
   \end{example}

    We also refer to \cite{gaudry2004construction} for
another example of resolution of a system of equations symmetric under
permutation of variables and to
 \cite{Colin97, gatermann1996semi} and references therein for more
 general questions (solving some systems invariant by action of some
 other groups).

 \paragraph{Related works.}  
A naive procedure to find $f_{\rm new}$ is to use
evaluation-interpolation method at enough points. Since the degree of
$f$ is $d$, then the numbers of momonials in $f$ and $f_{\rm new}$ are
at most $\rho := \binom{n+d}{n}$. Let $a_1, \dots, a_\rho$
be distinct points in $\K^n$. Given $(a_i, f_{\rm new}(a_i))_{1 \le i
  \le \rho}$, one can uniquely determine $f_{\rm new}$. To do it, we
let $b_1, \dots, b_\rho$ be  distinct points in $\K^n$ then we define
$a_i = u_i(b_i)$ for all $i=1, \dots, \rho$. Then $f_{\rm new}(a_i) =
f(b_i)$  for $i=1, \dots, \rho$. Interpolating at points  $(a_i,
f_{\rm new}(a_i))_{1 \le i  \le \rho}$ gives us the polynomial $f_{\rm
new}$. Note that although $b_i$'s are distinct, it is possible to have
some duplicates among $a_i$'s. For simplicity, we suppose that $(a_1,
\dots, a_\rho)$ are distinct. As mentioned above, we can assume that
$\deg(u_k) \le d$ for all $k=1, \dots, n$, then evaluating $u_k$ at
$b_i$ takes $O(\rho)$ operations in $\K$. So the total cost of
$O(n \rho^2)$ operations in $\K$ to find $(a_1, \dots,
a_\rho)$. Similarly, evaluating $f$ at $(b_1, \dots, b_\rho)$ requires
$O(\rho^2)$ operations in $\K$; then the same cost is needed to find
$(f_{\rm new}(a_i))_{1 \le i \le \rho}$. Finally, interpolating at
points  $(a_i, f_{\rm new}(a_i))_{1 \le i  \le \rho}$ can be done by
constructing a system of linear equations. Invertibility  of  the
Vandermonde matrix in $\K^{\rho \times  \rho}$ implies  that the
number of operations for this interpolation is $O(\rho^\omega)$. Here
$2 < \omega \le 3$ is the exponent of multiplying two square matrices
with coefficients in $\K$.  Thus the total  cost of $O(n \rho^\omega)$
is required to find $f_{\rm  new}$ by using the
evaluation-interpolation method.     

On the other hand, it is well known that one can use
{G}r{\"o}bner bases to obtain $f_{\rm new}$ from $f$ and $(u_1,
\dots, u_n)$, however, there is no complexity analysis for this
process. Precisely, in the polynomial ring $\K[x_1, \dots, x_n, e_1,
\dots, e_n]$, we fix a monomial order $\succ$ where any monomial
involving one of $(x_1, \dots, x_n)$ is greater than all monomials in
$\K[e_1, \dots, e_n]$. Let $B$ be a  {G}r{\"o}bner basis with respect
to the order $\succ$ of the ideal $\langle u_1-e_1, \dots,
u_n-e_n \rangle \subset \K[x_1, \dots, x_n, e_1, \dots, e_n]$. Let
$f$ be a polynomial in $\K[u_1, \dots, u_n]$. Then the
polynomial $f_{\rm new}$ can be obtained as the remainder of $f$ on
division by $B$. We refer the reader to \cite[Proposition~4 -
Section~1 - Chapter~7]{CLO07} for a detailed description of this
procedure.  

 Later in this paper, the main idea of our algorithm is to use  Hensel
 liftings.  This idea comes from the work given by  Bl\"aser and
 Jindal in~\cite{BlaserJindal18} when $G$  is the symmetric group
 $S_n$ and $u_1, \dots, u_n$ are elementary symmetric polynomials in
 $(x_1, \dots, x_n)$. The main idea of their algorithm is to use the
 fact that $x_i$ can be written as a function  of $u_i$ by using a
 polynomial $$q(t) = t^n - u_1t^{n-1} + \cdots + (-1)^nu_n \in \K[x_1,
 \dots, x_n][t],$$ where $t$ is a new variable. For example, consider
 $n=2$ and $(u_1, u_2) = (x_1+x_2, x_1 x_2)$. Then $x_1$ and $x_2$ are
 the roots of polynomial $p(t) = p^2 - (x_1+x_2)t + x_1x_2 =  t^2  -
 u_1 t + u_2$, and so $$x_1 =  \frac{u_1+\sqrt{u_1^2 - 
      4u_2}}{2} \, {\rm and } \, x_2 =    \frac{u_1-\sqrt{u_1^2 -
      4u_2}}{2}. $$ If we substitute these  functions to $f$ we
  obtain $f_{\rm new}$, however, these functions are  neither
  polynomials nor power series. In order to deal with this situation,
  they substitute $u_n$ by $u_n+(-1)^{n-1}$ (and replace
  $u_i$ by $e_i$ as we want to find the unique expression $f_{\rm 
    new} \in \K[e_1, \dots, e_n]$ of $f$) to obtain
  polynomial 
  \begin{equation}
      \label{eq:vite} q(\e, t) = t^n - e_1t^{n-1} + \cdots +
  (-1)^n(e_n+(-1)^{n-1}) 
  \end{equation} in  $\K[e_1, \dots, e_n][t]$,
   and then compute degree $d$ truncation of roots of $q(t)$, with
   respect to $t$, by using Newton's iteration. The substitution
   $u_n$ by $u_n+(-1)^{n-1}$ makes sure that the polynomial
   $q(0, \dots, 0,t)$ has no multiple root in order to perform
   Newton's iteration. Then together with the 
   fact that $q(e_1, \dots, e_n, t)$ is square-free with respect to
   $t$, one can conclude that the roots of $q(e_1, \dots, e_n, t)$ can
   be expanded into power series in $(e_1, \dots, e_n)$ (see
   e.g.~\cite[Condition~A]{sasaki1999solving}).
   
      Note that we can chose a polynomial $q  \in \K[\e, t]$ as 
\[
 q(\e, t) = t^n -
  (u_1(1, \dots, n)+e_1)t^{n-1} + \cdots +  (-1)^n (u_n(1,
  \dots,n)+e_n)
\]
   since $q(0, \dots, 0, t)$ has $n$
  distinct roots $1, \dots, n$ and $q(\e, t)$ is square-free with
  respect to $t$. We also remark that the algorithm given by Bl\"aser
  and Jindal in~\cite{BlaserJindal18} only works for the symmetric
  group ${S}_n$ and $(u_1, \dots, u_n)$ are elementary symmetric
  polynomials in $x_1, \dots, x_n$. This procedure can not be
  generalized for any finite reflection group or any set of generators
  since we can not always find a polynomial $q(\e, t)$ as in
  \eqref{eq:vite}; the equation~\eqref{eq:vite} is obtained thanks to
  Vieta's formulas.   

   Recently, a slight generalization of Bl\"aser  and Jindal's algorithm 
\cite{BlaserJindal18} 
and its complexity are described in \cite[Lemma~2.2]{FLSSV2021} for 
directed products of symmetric groups ${S} = {
  S}_{\ell_1}\times \cdots \times {S}_{\ell_r}$, where $\ell_k$'s
are positive integers, and the generators are also elementary
symmetric polynomials. To be specific, for $1 \le k \le r$, we denote by
$\x_k = (x_{k,1}, \dots, x_{k, \ell_k})$ a set of $\ell_k$ variables
and $\u_k = (u_{k, 1}, \dots, u_{k, \ell_k})$ the elementary
symmetric polynomials in $\x_k$. For any polynomial $f$ in $\K[\x_1,
\dots, \x_r]^{S}$, we want to find $f_{\rm new} \in \K[\e_1,
\dots, \e_r]$, where $\e_k = (e_{k, 1}, \dots, e_{k, \ell_k})$ are new
variables, such that $f_{\rm new}(\u_1, \dots, \u_r) =
f$. For $1 \le k \le r$, we define 
\[
 q_k(\e, t) = t^n - (u_1(1,
\dots, \ell_k)+e_{k,1})t^{n-1} + \cdots + 
 (-1)^n (u_{\ell_k}(1,
\dots,\ell_k)+e_{k, \ell_k}) \in \K[\u_k, t]
\]
 and then compute degree
$d$ truncation of roots of $q_k(\e, t)$, with respect to $t$, by using
Newton's iteration. 

Finally, there are also some works which study the relation between the
  sizes $L(f)$ and $L(f_{\rm new})$ of the smallest circuits computing
  $f$ and $f_{\rm new}$ respectively. In
  \cite[Theorem~1]{gaudry2006evaluation}, the authors show that $L(f)
  \leq \Delta(n) L(f_{\rm new})+2$, where $\Delta(n) \le 4^n(n!)^2$,
  when the group $G$ is the symmetric group. Later on, the results in
  \cite[Theorem~3]{dahan2009evaluation} are in a more general setting
  when the polynomials are invariant under the action of general
  finite matrix groups. While the runtime of the algorithm in
  \cite{dahan2009evaluation} depends on the order of the group, the
  runtime of our algorithm in this paper does not depend on the order
  of the group. However, the result in  \cite{dahan2009evaluation} is
  for any  finite groups, ours here is for finite reflection groups. 
    
    We will consider the problem of computing $f_{\rm new}$ when $G$
    is generally a  finite group as one of future works. Note that,
    when $G$ is a finite group, any polynomial $f \in \K[x_1, \dots,
    x_n]^G$, can be uniquely written as  
  \[
  f = \sum_i \theta_ih_i(e_1, \dots, e_n)
  \] for some $h_i$ in $\K[y_1, \dots, y_n]$, where $(e_1, \dots, 
  e_n)$ and $(\theta_1, \dots,  \theta_N)$ are respectively primary
  and minimal secondary invariants.

\paragraph{Our result.} In what follows, we use $\softO(\cdot)$ to
indicate that polylogarithmic factors are omitted, that is, $a$ is in
$\softO(b)$ if there exists a constant $k$ such that $a$ is $O(b \,
\log^k(b))$ (see~\cite[Section 7 - Chapter 25]{Gat03} for technical
details). The smallest integer larger or equal to $a$ is written as
$\lceil a \rceil$.

          For a positive integers $D$ and $m$,  ${\mathcal{M}}_b(D)$
        denotes the cost of the 
        multiplication of univariate polynomials of degree $D$ in
        terms of operations in the base ring $\K$ and ${\mathcal{M}}(D, m)$
        denotes the cost of $m$-variate series multiplication at
        precision $D$. The quality ${\mathcal{M}}_b(D)$ can be taken
        in $\softO(D)$ by using the  algorithms of Sch{\"o}nhage and
        Strassen \cite{schonhage1971schnelle} and Sch{\"o}nhage 
        \cite{schonhage1977schnelle} and ${\mathcal{M}}(D, m)$ can be taken
        less than ${\mathcal{M}}_b((2D+1)^m)$  by using  Kronecker's
        substitution (see \cite{Kronecker82} and \cite{Gat03}). If the
        field $\K$ is of characteristic zero,  ${\mathcal{M}}(D, m)$ is
        $\softO({\mathcal{M}}_b(\binom{D+m}{m}))$, that is linear in the
        size of the series, up to logarithmic factors (see
        \cite{lecerf2003fast}).

We suppose that the sequence of polynomials $(u_1, \dots, u_n)$ is given
by a straight-line program, that is, a sequence of elementary
operations $+, -, \times,$ to represent $(u_1, \dots, u_n)$. 

\begin{theorem} Let $\K[x_1, \dots, x_n]$ be a multivariate polynomial
  ring over a field $\K$. Let $(u_1, \dots, u_n)$ be algebraically
  independent elements in $\K[x_1, \dots, x_n]$. Then there exists an
  algorithm called {\sf Convert\_Polynomial} 
  which, takes as input $(u_1, \dots, u_n)$ and a polynomial $f$ in
  $\K[u_1, \dots,  u_n]$, and outputs the polynomial $f_{\rm new} \in
  \K[e_1, \dots, e_n]$ such that $f_{\rm  new}(u_1, \dots, u_n) =  
  f(x_1, \dots, x_n)$ using $$\softO\big((nL  + n^4){\mathcal{M}}(d,
  n) + \binom{n+d}{n}^2\big)$$ operations in $\K$, where $d$ is  the 
  degree of $f$ and $L$ is the length of a straight-line program 
  representing $(u_1, \dots, u_n)$.  
\end{theorem}

Note that the straight-line program coding for the input of our
algorithm is not restrictive since the notion of straight-line
program encoding covers dense encoding notion.  Precisely, if $h$ is a  
polynomial of degree $d$ in $\K[x_1, \dots, x_n]$, then the length of
a  straight-line program representing $h$ is
$O\big(\binom{n+d}{n}\big)$. This can be  seen as follows: the number
of monomials of degrees at most $d$ in $\K[x_1, \dots, x_n]$ is
$\binom{n+d}{n}$, taking the multiplication of all monomials of $h$
with their coefficients and adding them  up requires $2\binom{n+d}{n}$
operations.  

 As mentioned above, while there is no complexity analysis for the
 process using {G}r{\"o}bner bases, when the field $\K$ is large
 enough,  e.g. of characteristic zero, the number of required
 operations is  $O(n \binom{n+d}{n}^w)$ for the procedure using
 evaluation-interpolation method. Therefore, our algorithm in this
 paper works for a general problem and is faster compared to
 previous ones. 
  In addition our algorithm has been implemented in the Maple computer
  algebra and has been tested  when $\K$ is the field of rational
  numbers and $G$ are symmetric groups, hyperoctahedral groups, and
  the symmetric group $D_3$ of regular 3-gon in $(x_1, x_2)$-plane. 

\paragraph{Organization.} The structure of the paper is as follows. In
Section~\ref{sec:lifting}, we provide a detailed description of a
lifting procedure and its complexity to compute an approximation of a
vector of multivariate power series. Our main algorithm and
its cost are given in Section \ref{sec:main}. 


\section{Newton-Hensel lifting}
\label{sec:lifting}
  Lifting techiques are classical methods which can be found, for
  examples, as in~\cite{heintz2000deformation, giusti2001grobner,
    schost2003computing} (see also references therein). 
   In this section, the notations $\x$ and $\e$ stand respectively for
   the sets of variables $(x_1, \dots, x_n)$ and $(e_1, \dots, e_m)$
   for some positive integers $n$ and $m$
   and $\langle \e \rangle$ is an ideal in the polynomial ring 
   $\K[\e]$ generated by $e_1, \dots, e_m$. For a positive integer
   $d$, let us denote by $$\langle \e\rangle^d := \langle e_1^{t_1}
   \cdots e_m^{t_m} \, : \, t_1+ \cdots + t_m \ge d+1\rangle  \subset
   \K[\e] $$ an ideal in $\K[\e]$ generated by  all monomials of
   degrees at least $d+1$. If $p$ is a polynomial in $\K[\e]$ or a
   power series in $\K[[\e]]$, $p$ mod $\langle \e \rangle^d$ equals
   part of $p$ up to degree $d$.

   \begin{proposition} \label{prop:newton}
   Let $\h =(h_1, \dots, h_n)$ be a sequence of polynomials in $\K[\x,
   \e]$ and $\balpha = (\alpha_1, \dots, \alpha_n)$ in $\Kbar{}^n$
   be a solution of $\h(\x, 0, \dots, 0)$. Assume  that the
   Jacobian matrix of $\h(\x, 0, \dots, 0)$ with respect to $\x$ is
   full rank at $\balpha$. Then  there exists a unique vector of power
   series $\v = (v_1, \dots, v_n)$ in $\Kbar[[e_1, \dots, e_m]]$ such
   that      
   \begin{equation}\label{eq:satis}
   \v(0, \dots, 0) = \balpha \quad {\rm and} \quad h_1(\v, \e) =
   \cdots =  h_n(\v, \e) = 0. 
   \end{equation}
   Furthermore, the vector of power series $\v$ can be approximated
   to arbitrary degree $\delta$ using ${\sf Lifting}$ algorithm, that
   is, 
   \[
   \v^{(\lceil \log_2(\delta)\rceil)} = \v \mod \langle \e
   \rangle^\delta,
   \] where $\v^{(\lceil \log_2(\delta)\rceil)}$ is the output of
   ${\sf Lifting}$ algorithm which takes $\h, \balpha$, and $\delta$
   as the input. The complexity to compute this approximation is
   $$\softO((nL+ n^4){\mathcal{M}}(\delta, m))$$ operations in $\K$, where
   $L$ is length of a straight-line program computing $\h$.
   \end{proposition}

\begin{algorithm}[h] 	 
  \caption{${\sf Lifting}(\h, \balpha, \delta) $}

~\\

  {\bf Input:} a sequence of polynomial $\h = (h_1, \dots, h_n)$ in
  $\K[\x,\e]^n$, with $\x = (x_1, \dots, x_n)$ and $\e = (e_1, \dots,
  e_m)$, a point $\balpha = (\alpha_1, \dots, \alpha_m)$ in
  $\Kbar{}^m$, a positive integer $\delta$, such that 
  \begin{itemize}
    \item $\balpha$ is a root of $\h(\x, 0, \dots, 0)$ 
    \item the Jacobian matrix of $\h(\x, 0, \dots, 0)$ with respect to 
      $\x$ has full rank at $\balpha$ 
    \end{itemize}

  {\bf Output:} the vector approximates $\balpha$ in $\K[[\e]]$ with
  precision $\delta$
 \begin{enumerate}
\item set $\v^{(0)} = \balpha$
\item compute the Jacobian matrix $\jac$ of $\h$ with respect to $\x$
\item for $k$ from $1$ to $\lceil \log_2(\delta) \rceil$ do: 
 \begin{enumerate}
\item compute $$\v^{(k)} = \left(\begin{matrix}v^{(k-1)}_1 \\ \vdots \\
      v^{(k-1)}_n \end{matrix}\right) - \left(\jac(\v^{(k-1)}, \e
    )\right)^{-1} \, \left(\begin{matrix}h_1(\v^{(k-1)}, \e)
      \\ \vdots \\ h_n(\v^{(k-1)},\e) \end{matrix}\right)$$
\end{enumerate}
\item return $\v^{(\lceil \log_2(\delta) \rceil)}$
\end{enumerate}
  \label{alg:lift} 	 
\end{algorithm} 
\begin{example}
    Let us consider polynomials $h_1 = x_1+x_2-e_1-2$ and $h_2 =
    x_1^2+x_2^2 -e_2-10$ in $\K[x_1, x_2, e_1, e_2]$. The point
    $(-1,3)$ is a root of $h_1(x_1,x_2,0,0) = h_2(x_1,x_2,0,0) =
    0$ and the Jacobian matrix of $(h_1(x_1, x_2, 0, 0), h_2(x_1, x_2,
    0, 0))$  with respect to $(x_1, x_2)$ 
    has full rank at $(-1,3)$.    
     The power series $$v_1 = -1 + \frac{3}{4}e_1 -\frac{1}{8}e_2 +
    \frac{5}{64}e_1^2 - \frac{1}{64}e_1e_2 + \frac{1}{256}e_2^2 +
    \langle e_1, e_2\rangle^2$$ and $$v_2 = 3+ \frac{1}{4}e_1 +
    \frac{1}{8}e_2 - \frac{5}{64}e_1^2 + \frac{1}{64}e_1e_2 -
    \frac{1}{256}e_2^2 +  \langle e_1, e_2\rangle^2$$ in $\K[[e_1,
    e_2]]$ satisfy $v_1(0,0) = 1$, $v_2(0,0) = 3$, and
    $h_1(v_1,v_2,e_1,e_2) = h_2(v_1,v_2,e_1,e_2) = 0$.   
\end{example}

    The rest of this section is devoted to prove
    Proposition~\ref{prop:newton}. Let us denote by $\jac$ the
    Jacobian matrix of $\h$ with respect to 
   $\x$. We first prove that the sequence $\left(\v^{(k)}\right)_{k \ge 0}$
    is well-defined. To do it, we prove the following claims. 
    \begin{lemma} For any integer $k \in \ZZ_{\ge 0}$, 
      \begin{itemize}
        \item[$(${\bf a}$)$] the determinant $\jac$ at
          $\v^{(k)}$ is invertible in $\K[\e]/\langle \e \rangle $ and  
          \item[$(${\bf b}$)$] $h_i(\v^{(k)}, \e) = 0$ mod $\langle \e
            \rangle^{2^k}$ for all $i=1, \dots,n$. 
        \end{itemize}  \label{lemma:convergent}    
      \end{lemma}

 \begin{proof}
     We prove these claims by induction on $k$. For $k = 0$, since
     $\v^{(k)}$ equals $\balpha$, the claims follow from the
     assumptions saying  that $\balpha$ is a root of $\h(\x, 0, \dots,
     0)$ and  the Jacobian matrix of $\h(\x, 0, \dots, 0)$ with
     respect to $\x$ has full rank at $\balpha$.   

     Let us assume that both $(${\bf a}$)$ and $(${\bf b}$)$ hold for
     $k \ge 0$. We will show that $(${\bf a}$)$ and $(${\bf b}$)$ also
     hold for $k+1$. First we have 
     \begin{equation} \label{eq:jac}
     \v^{(k+1)}-\v^{(k)} = -\jac(\v^{(k)},
     \e)^{-1}\left(\begin{matrix}h_1(\v^{(k)}, \e) 
      \\ \vdots \\ h_n(\v^{(k)},\e) \end{matrix}\right). 
     \end{equation} Together with the induction hypothesis for $k$,
     one can deduce 
     that \begin{equation}\v^{(k+1)}-\v^{(k)} = 0 \ \mod \langle \e
            \rangle^{2^k}. \label{eq:minus}\end{equation}

     Let us denote by $J$ the determinant of $\jac$, the Jacobian
     matrix of $\h$ with respect to $\x$. Then by Taylor expansion, 
\[
 J(\v^{(k+1)}, \e) = J(\v^{(k)}, \e) + \sum_{i=1}^n\frac{\partial
       J}{\partial x_i}(\v^{(k)}, \e) (\v^{(k+1)} - \v^{(k)})  \mod
     \langle \v^{(k+1)} - \v^{(k)}\rangle^2,
\]
    where $\langle \v^{(k+1)} - \v^{(k)}\rangle = \langle
     v_1^{(k+1)} - v_1^{(k)}, \dots,  v_n^{(k+1)} - v_n^{(k)} \rangle$
     an ideal in $\Kbar(\e)$. Moreover we have $J(\v^{(k)}, \e)$ is
     non-zero in $\K[\e]/\langle \e \rangle$ by the induction
     hypothesis for $k$ and $\v^{(k+1)}-\v^{(k)} = 0 \ \mod \langle \e
            \rangle$ by~\eqref{eq:minus}. Therefore $J(\v^{(k+1)}, \e)
            \ne 0 \mod \langle \e \rangle$,  which gives our claim
            $(${\bf a}$)$ for 
            $k+1$. 

     To prove part $(${\bf b}$)$ holds for $k+1$, we first multiply
     both sides of~\eqref{eq:jac} with $\left( \frac{\partial
         h_i}{\partial x_1}, \dots, \frac{\partial
         h_i}{\partial x_n}\right)$, for $i=1, \dots, n$, to have 
     \[
     \left( \frac{\partial
         h_i}{\partial x_1}, \dots, \frac{\partial
         h_i}{\partial x_n}\right) \left(\v^{(k+1)}-\v^{(k)}\right)^T
     = h_i(\v^{(k)}, \e). 
     \] In addition, by Taylor expansion of $h_i$ between the points
     $\v^{(k+1)}$ and $\v^{(k)}$,  one has 
\[
           h_i(\v^{(k+1)}, \e) = h_i(\v^{(k)}, \e) + \sum_{j=1}^n
     \frac{\partial h_i}{\partial x_j}h_i(\v^{(k)},
     \e)(\v^{(k+1)}-\v^{(k)})   \mod 
     \langle \v^{(k+1)} - \v^{(k)}\rangle^2. 
\]
      These two facts imply that $$h_i(\v^{(k+1)}, \e) = 0 \mod
     \langle\v^{(k+1)} - \v^{(k)}\rangle.$$ Then together
     with~\eqref{eq:minus}, one can conclude that $h_i(\v^{(k+1)}, \e)
     = 0 \mod \langle \e \rangle^{2^{k+1}}$, which is our claim
     $(${\bf b}$)$ for $k+1$. 
     \end{proof}
 
  Let us  conclude the existence of the vector of power series
     $\v$ in $\Kbar[[\e]]^n$ which satisfies~\eqref{eq:satis}. As
     in~\eqref{eq:minus}, the equation $\v^{(k+1)}-\v^{(k)} = 0 \ \mod
     \langle \e \rangle^{2^k}$ holds for any $k\in \ZZ_{\ge 0}$. Then
     for any $i=1, \dots, n$, the sequence of functions 
            $\left(v_i^{(k)}\right)_{k\in \ZZ_{\ge 0}}$ converges to a
            power series $v_i$ in $\K[[\e]]$. Furthermore, by
            Lemma~\ref{lemma:convergent}$(${\bf b}$)$, $$h_i(\v^{(k)},
            \e) = 0 \mod \langle \e\rangle^{2^{k}}  {\rm for~} i=1, \dots, n
            {\rm ~and~all~} {k\in \ZZ_{\ge 0}},$$ which implies that $h_i(\v,
            \e) =0$ holds in $\Kbar[[\e]]$ for all $i=1, \dots,
            n$. Finally, the relation in~\eqref{eq:minus} also gives
            us $\v(0, \dots, 0) = \balpha$ .

  We now investigate the complexity to compute
    $\v^{\lceil\log_2(\delta)\rceil}$  from $\balpha$ and $\h$, which
    finishes our proof for  Proposition~\ref{prop:newton}.   
    \begin{lemma} Let $L$ be the complexity to compute $\h = (h_1,
      \dots, h_n)$. Then the
      complexity to compute the approximation $\v^{\lceil\log_2(\delta)\rceil}$ of 
      $\balpha$ in $\Kbar[[\e]]$ with precision $\delta$ is $O((nL
        + n^4){\mathcal{M}}(\delta, m))$ operations in $\K$. 
      \end{lemma}

\begin{proof} The complexity to compute the first partial
        derivatives $\left(\frac{\partial h_i}{\partial x_j}\right)_{1
          \le i,j\le n}$ of $\h$ is $O(nL)$ operations in $\K$ by
        \cite[Theorem~1]{Baur1983complexity} or
        \cite[Lemma~25]{giusti1997lower}. Therefore one needs the same
        cost to compute the Jacobian matrix $\jac$ of $\h$ with
        respect to $\x$. 

        Given the approximation $\v^{(k)}$, we evaluate the complexity
        to compute  $\v^{(k+1)}$. Evaluating the matrix  $\jac$ and the
        vector $\h$ at  $\v^{(k)}$ takes $O(nL)$ operations in
        $\K[[e_1, \dots, e_m]]/\langle \e\rangle^{t^{2^{k+1}}}$ by
        using Baur-Strassen's 
        algorithm~\cite{Baur1983complexity}. Besides, computing the
        inversion of 
        $\jac(\v^{(k)}, \e)$ requires $O(n^4)$
        operations in  $\K[[e_1, \dots, e_m]]/\langle
        \e\rangle^{t^{2^{k+1}}}$ by using, for instance, Leverrier's
        algorithm~\cite{le1840variations}. Finally, the cost of an
        operation in  the
        quotient ring $\K[[e_1, \dots, e_m]]/\langle
        \e\rangle^{t^{2^{k+1}}}$  is $\mathcal{M}(2^{k+1}, m)$ operations
        in $\K$. Therefore, given $\v^{(k)}$, the
        total cost to  compute $\v^{(k+1)}$ is $O((nL +
        n^4){\mathcal{M}}(2^{k+1}, m))$ operations in $\K$. 
        
        Thus the total cost to compute the approximation
        $\v^{\lceil\log_2(\delta)\rceil}$ of  $\balpha$ is $$O((nL +
        n^4)\sum_{k=0}^{\lceil
          \log(\delta)\rceil}{\mathcal{M}}(2^{k+1}, m)) = O((nL 
        + n^4){\mathcal{M}}(\delta, m))$$ operations in $\K$.
\end{proof}


\section{The main algorithm} \label{sec:main}

Let $G$ be a finite reflection group and $(u_1, \dots, u_n)$ be a set
of generators of $\K[\x]^{G}$ with $\x = (x_1, \dots, x_n)$. Let $f$
be a polynomial in $\K[\x]^{G}$ of degree $d$. In this section, we
present our algorithm and its complexity to compute the polynomial
$f_{\rm new}$ in $\K[\e]$, with $\e = (e_1, \dots, e_n)$, such that
$f_{\rm new}(u_1, \dots, u_n) = f$.

Our main idea is to eliminate the variables $\x$ from the system
$\bar \u = (\bar u_1, \dots, \bar u_n) = (u_1-e_1, \dots, u_n - e_n) \in \K[\x, \e]^n$ by using
linear algebra performance, i.e., the Hensel-Newton's lifting to be more
precise.  

Assume there exists a point $\balpha \in \Kbar{}^n$ such that
$\balpha$  is a solution of $\bar \u(\x, 0, \dots, 0)$ and the Jacobian
matrix of $\bmu$ with respect to  $\x$ has full rank at
$\balpha$. Then, by Proposition~\ref{prop:newton},  there exists a
unique vector of power series $\v = (v_1, \dots, v_n)\in \K[[\e]]^n$
such that  
  \begin{equation*}
   \v(0, \dots, 0) = \balpha \quad {\rm and} \quad \bar u_1(\v, \e) =
   \cdots =  \bar u_n(\v, \e) = 0. 
   \end{equation*} In order to find $f_{\rm new}$  we only need
   truncations of $(v_1, \dots, v_n)$ at precision $d$, which can be
   done by using ${\sf Lifting}$ algorithm. We then finally evaluate
   $f$ at these truncated power series to obtain $f_{\rm new}$.

However, for a root $\balpha$ of $\bar \u(\x, 0, \dots, 0)$, the
  Jacobian matrix of $\bar \u(\x, 0, \dots, 0)$ with respect to $\x$ is
  not always full rank at $\balpha$. 
 \begin{example} Let us consider $n=3$ and $G$ be the symmetric group
    $S_3$. In this case, we can take $(u_1, u_2, u_3) = (x_1+x_2+x_3,
    x_1^2+x_2^2+x_3^2, x_1^3+x_2^3+x_3^3)$ the power sum symmetric
    functions of $(x_1, x_2, x_3)$. Then $\bar \u = (x_1+x_2+x_3-e_1,
    x_1^2+x_2^2+x_3^2-e_2, x_1^3+x_2^3+x_3^3-e_3)$ and the Jacobian
    matrix of $\bar \u$ with respect to $(x_1, x_2, x_3)$ is 
    $$\jac = \left( \begin{matrix} 1 & 1 &1 \\ 2x_1 &
        2x_2 &2x_3 \\ 3x_1^2 & 3x_2^2 &
        3x_3^2\end{matrix}\right).$$ The point $\balpha = (0,0,0)$ is a 
    solution of $\bar \u(x_1, x_2, x_3, 0, 0,0)$, however, the rank of Jac
    at $\balpha$ is equal to 1. \label{ex:not_full} 
  \end{example}

  In order to deal with the above
  situation, we take a random point  $\balpha \in \Kbar{}^n$, then we
  compute a new polynomial system $$\u := (u_1-e_1-c_1, \dots,
  u_n - e_n-c_n) \in \K[\x, \e],$$ where $c_i = u_i(\balpha)$
  for $1 \le i \le 
  n$. By this way, the point $\balpha$ is a root of $\u(\x, 0, \dots,
  0)$ and the Jacobian matrix of $\u(\x, 0, \dots, 0)$ with respect
  to $\x$ is full rank at $\balpha$ by Lemm~\ref{lemma:shift} 
  below. 
  \begin{example}\label{ex:full}
    Let us consider $G$ and $(u_1, u_2, u_3)$ as in
    Example~\ref{ex:not_full}. We take a random point $\balpha = (4,
    6, 0)$ in $\Kbar{}^3$. The Jacobian matrix  of $\u  =
    (x_1+x_2+x_3-e_1- 10, x_1^2+x_2^2+x_3^2-e_2 - 52,
    x_1^3+x_2^3+x_3^3-280)$ with respect to $(x_1, x_2, x_3)$ has full
    rank $3$ at  $\balpha$.     
  \end{example}

 Then if  $\v^{(\lceil \log_2(d)\rceil)}$ is the truncation of
  the vector of power series $\v$ and  $\bar f_{\rm new}$, which
  indeed equals to $f_{\rm new}(e_1 -c_1, \dots, e_n-c_n)$, is the
  evaluation of $f$ 
   at $\v^{\lceil \log_2(d)\rceil}$, we then  apply the translation
   $(e_i)_{1 \le i \le n} \leftarrow (e_i + c_i)_{1 \le i \le n}$ in
   order to obtain the polynomial $f_{\rm new}$. The result of this
   process is a so called ${\sf Convert\_Polynomial}$ algorithm. 

\begin{example} Continuing with Example~\ref{ex:full} and considering $f =
  x_1^3+x_2^3+x_3^3-2x_1x_2x_3 - x_1-x_2-x_3$ a polynomial in $\K[x_1,
  x_2, x_3]^{{S}_3}$ of degree $d=3$, the ${\sf
    Lifting}$ procedure takes $\bmu$, 
 $(4,6,0)$, and $d$ as the 
  input, and outputs
\begin{align*}
    v_1^{(2)} &= \frac{-1}{512}e_1e_2e_3 - \frac{11}{96}e_1^3
  +\frac{33}{4096}e_2^3-\frac{1}{55296}e_3^3+ \cdots  + \frac{3}{8} e_2
        - \frac{1}{24} e_3+ 4, \\
    v_2^{(2)} &=  \frac{11}{2592}e_1e_2e_3 + \frac{1}{24}e_1^3
  -\frac{197}{31104}e_2^3=\frac{29}{1679616}e_3^3 + \cdots  
-      \frac{1}{6} e_2 + \frac{1}{36} e_3
      + 6, {\rm ~and}\\
    v_3^{(2)} &=\frac{-95}{41472}e_1e_2e_3 +\frac{7}{96}e_1^3
  -\frac{1715}{995328}e_2^3 +  \frac{11}{13436928}e_3^3+ \cdots   + e_1
      - \frac{5}{24}e_2 + \frac{1}{72}e_3.
\end{align*}
Then we substitute $(x_1, x_2, x_3) = (v_1^{(2)}, v_2^{(2)}, v_3^{(2)})$
into $f$ and truncate the result at degree $3$ to obtain $$\bar{f}_{\rm
  new} = -1/3e_1^3 -10e_1^2 + e_1e_2 -49e_1+10e_2 + 1/3e_3 + 270.$$
Finally $${f}_{\rm new} =  \bar{f}_{\rm new}(e_1-10, e_2-52, e_3-280)
= -1/3e_1^3  + e_1e_2 -e_1+1/3e_3.$$
\end{example}
   
To
   conclude the correctness of our algorithm, we
    need the following result in order to conclude the correctness 
   of the algorithm.  

\begin{lemma} Let $\h = (h_1, \dots, h_n)$ be a sequence of
  polynomials in $\K[\x]$, with $\x = (x_1, \dots, x_n)$, and $\balpha
  = (\alpha_1, \dots, \alpha_n)$ be a  random point in $\Kbar{}^n$. We
  define a system of polynomials $\f = (f_1, \dots, f_n) = (h_1 -
  h_1(\balpha), \dots,  h_n - h_n(\balpha))$ in $\K[\x]$. 
   Then the Jacobian matrix of $\f$ with respect to $\x$ has full rank
   at $\balpha$. \label{lemma:shift} 
\end{lemma}

\begin{proof} 
To prove our claim, it suffices to show that there exits a
non-empty Zariski open set $\mathscr{U} \subset \Kbar{}^n$ such that
for any $\balpha \in \mathscr{U}$, the Jacobian matrix of $\f$ with
respect to $\x$ has full rank at $\balpha$. 

Let $\z = (z_1, \dots, z_n)$ be a set of new variables. We set
$\mathfrak{f} =  (h_1-z_1,
\dots, h_n-z_n)$ a polynomial system in $\K[\x, \z]$. For a point
$\rho = (\rho_1, \dots, \rho_n) \in \Kbar{}^n$, we denote by
$\Theta_\rho$ the mapping  
\begin{align*}
  \Theta_\rho :  \K[\z][x_1, \dots, x_n] &\rightarrow \Kbar[x_1,
                                           \dots, x_n] \\ 
   \mathfrak{f}  &\mapsto (h_1-\rho_1, \dots, h_n-\rho_n)
\end{align*}by setting $\z$ equals $\rho$.  Note that, for any $\rho
\in \Kbar{}^n$,  the Jacobian matrix of $\Theta_\rho(\mathfrak{f} )$
with respect to $\x$ equals that of $\f$. 

Consider the mapping 
\begin{align*}
 \Theta :  (\balpha, \rho) \in \Kbar{}^n \times \Kbar{}^n &\rightarrow
                                                            \Theta_\rho(\mathfrak{f})(\balpha).           
\end{align*} Since the columns corresponding to partial derivatives of
$\mathfrak{f}$ with respect to $\z$ contain a $\diag(-1, \dots, -1)$
matrix, then the Jacobian matrix of $\mathfrak{f}$  has full rank at
all points $(\balpha, \rho)$ of its zero-set. In other words, $\bf 0$
is a regular value of $\Theta$. 

Then, by Thom's weak transversality theorem (see
e.g.~\cite[Proposition B.3]{din2017nearly} for the algebraic version),
there exists a non-empty Zariski open set $\mathscr{O} \subset
\Kbar{}^n$  such that for $\rho \in \mathscr{O}$, $\bf 0$ is a 
regular value of the induced mapping
\begin{align*}
     \balpha \in \Kbar{}^n  \rightarrow
  \Theta_\rho(\mathfrak{f})(\balpha). 
\end{align*} That is, the Jacobian matrix of
$\Theta_\rho(\mathfrak{f})$ has rank $n$ at any root $\balpha \in
\Kbar{}^n$ of $\Theta_\rho(\mathfrak{f})$.  

As a consequence, for the dense Zariski open set $\mathscr{U} :=
\h^{-1}(\mathscr{O})$, we have for all points $\balpha \in
\mathscr{U}$, the Jacobian of $\h$ evaluated at $\balpha$
has full rank. In addition, $\balpha$ is a root of $\f$ and the
Jacobian of $\h$ is the same as that of $\f = \h -
\h(\balpha)$. Altogether, we obtain  our claim.
\end{proof}

\begin{algorithm}[h] 	 
  \caption{${\sf Convert\_Polynomial}\left((u_1, \dots,
      u_n), f\right)$} 
~\\

  {\bf Input:}   $n$ algebraically independent elements $(u_1, \dots,
  u_n)$ in $\K[x_1, \dots, x_n]$ and a polynomial $f \in \K[u_1, \dots, u_n]$ of
  degree $d$\\
  
 {\bf Output:} the polynomial $f_{\rm new} \in \K[e_1,
 \dots, e_n]$ 
  such that  $f_{\rm new}(u_1, \dots, u_n) = f$
 \begin{enumerate}
\item take a random point $\balpha = (\alpha_1, \dots, \alpha_n)$ in
  $\Kbar{}^n$  
\item \label{step:c}compute $(c_1, \dots, c_n) = (u_1(\balpha),
  \dots, u_n(\balpha))$ 
\item define polynomials $\u = (u_1-e_1-c_1, \dots,
  u_n-e_n-c_n)$ in $\K[x_1, \dots, x_n, e_1, \dots, e_n]^n$
\item compute $(v_1, \dots, v_n) = {\sf Lifting}(\u, \balpha, d) \in
  \K[[e_1, \dots, e_n]]^n$ \label{step:l}
\item find $f(v_1, \dots, v_n)$ and truncate the result at degree
  $d$ to obtain $\bar f_{\rm new}$ \label{step:e}
\item  return $\bar f_{\rm new}(e_1+c_1, \dots, e_n +c_n)$
\end{enumerate}
  \label{alg:a} 	 
\end{algorithm}  



The correctness of the algorithm is obtained from the
  above discussion. It remains to establish a complexity analysis of
  our algorithm. 

  Evaluating $u_k$ at $\balpha$ takes  $O(L)$ operations in $\K$;
  so the total cost of $O(n L)$ operations in $\K$ to compute 
  $c_1, \dots, c_n$ at Step~\ref{step:c}. At the core of the
  algorithm, that is at Step~\ref{step:l}, we need $$\softO((nL
        + n^4){\mathcal{M}}(d, n))$$ operations in $\K$ to compute the
        truncated power series $(v_1, \dots, v_n)$ by using ${\sf
          Lifting}$ algorithm from Proposition~\ref{prop:newton}. 
 We then evaluate $f$ at these truncated power series at
 Step~\ref{step:e}. Since $f$ has degree $d$, this can be done by
 using  $O\big(\binom{n+d}{n}\big)$ operations on $n$-variate power series
 truncated in $d$, for a total of $O\big(\binom{n+d}{n}^2\big)$
 in $\K$. This step gives  $$\bar f_{\rm new} = f_{\rm new}(e_1
 -c_1, \dots, e_n-c_n).$$ 

We  finally apply the translation $(e_i)_{1 \le i \le n}
 \leftarrow (e_i + c_i)_{1 \le i \le n}$ in order to obtain the polynomial
 $f_{\rm new}$. To do it, we incrementally compute the translates of all
   monomials of degree up to $d$ through successive multiplications
   and then, before combining, using  the coefficients of $\bar f_{\rm
     new}$. This step requires  $O\big(\binom{n+d}{n}^2\big)$
   operations in $\K$. 
 Therefore, the total cost to compute the polynomial $f_{\rm new}$
 is $${\softO\big((nL  + n^4){\mathcal{M}}(d, n) +
   \binom{n+d}{n}^2\big)}$$  operations in $\K$, as required. 

\bibliographystyle{plain}
\bibliography{biblio} 

\end{document}